\documentclass[journal]{IEEEtran}    

\usepackage{wrapfig}
\let\labelindent\relax
\usepackage{enumitem,soul}
\usepackage{rotating}
\usepackage{chngcntr}
\usepackage{apptools}
\usepackage{listings}
\usepackage{graphicx}
\usepackage{graphics}
\usepackage{amssymb}
\usepackage{amsmath}
\usepackage{color,cite}
\usepackage{soul}
\usepackage{xspace}
\usepackage{algpseudocode}
\usepackage{algorithm,algorithmicx}
\usepackage{bbm}
\usepackage{comment}
\usepackage{subfig}
\usepackage{psfrag}
\usepackage{tikz}
\usetikzlibrary{math}
\usetikzlibrary{shapes,arrows}
\usetikzlibrary{arrows.meta}
\tikzset{>={Latex[width=2.5mm,length=2.5mm]}}
\usetikzlibrary{positioning}
\tikzstyle{block}=[draw opacity=0.7,line width=1.4cm]
\usepackage{rotating}
\usepackage{graphicx}
\usepackage{sidecap}

\title{\LARGE \bf
A sub-modular receding horizon solution for\\ mobile multi-agent persistent monitoring}

\author{Navid Rezazadeh and Solmaz S. Kia
  \thanks{This work is supported by NSF award IIS-SAS-1724331. A preliminary version of this work will appear in the proceeding of the 8th IFAC Workshop on Distributed Estimation and Control in Networked Systems~\cite{NR-SSK:19}.} %
}

\newcommand{\VV}{\mathcal{V}}

\newcommand{\real}{{\mathbb{R}}} 
 
\newcommand{\realpositive}{{\mathbb{R}}_{>0}}
\newcommand{\realnonnegative}{{\mathbb{R}}_{\ge 0}}

\renewcommand{\baselinestretch}{.985}

\newcommand{\argmax}{\operatorname{argmax}}
\newcommand{\vect}[1]{\boldsymbol{\mathbf{#1}}}
\newcommand{\vectsf}[1]{\vect{\mathsf{#1}}}

\newcommand{\SUM}[2]{\sum_{#1}^{#2}}
 \newcommand{\boxend}{\hfill \ensuremath{\Box}}

\newtheorem{thm}{Theorem}[section]

\newtheorem{cor}{Corollary}[section]
\newtheorem{lem}{Lemma}[section]

\newtheorem{assump}{Assumption}

\makeatletter
\renewcommand*{\@opargbegintheorem}[3]{\trivlist
      \item[\hskip \labelsep{\bfseries #1\ #2}] \textbf{(#3)}\ \itshape}
\makeatother

\newcommand{\oprocendsymbol}{\hbox{$\bullet$}}
\newcommand{\oprocend}{\relax\ifmmode\else\unskip\hfill\fi\oprocendsymbol}



\begin{document}
\maketitle
\thispagestyle{empty}
\pagestyle{empty}

\begin{abstract}                          
We consider persistent monitoring of a finite number of inter-connected geographical nodes by a group of~heterogeneous mobile agents.  We assign to each geographical node a concave and increasing reward function that resets to zero after an agent's visit. Then, we design the optimal dispatch policy of which nodes to visit at what time and by what agent by finding a policy set that maximizes a utility that is defined as the total reward collected at visit times. We show that this optimization problem is NP-hard and its computational complexity increases exponentially with the number of the agents and the length of the mission horizon. By showing that the utility function is a monotone increasing and submodular set function of agents’ policy, we propose a suboptimal dispatch policy design with a known optimality gap. To reduce the time complexity of constructing the feasible search set and also to induce robustness to changes in the operational factors, we perform our suboptimal policy design in a receding horizon fashion. Then, to compensate for the shortsightedness of the receding horizon approach we add a new term to our utility, which provides a measure of nodal importance beyond the receding horizon.
This term gives the policy design an intuition to steer the agents towards the nodes with higher rewards on the patrolling graph. Finally, we discuss how our proposed algorithm can be implemented in a decentralized manner. A simulation study demonstrates our~results.
\end{abstract}

\section{Introduction}
\vspace{-0.03in}
In recent years, coordinating the movement of mobile sensors to cover areas that have not been adequately sampled/observed has been explored in controls, wireless sensors and robotic communities with problems related to coverage, exploration, and deployment. Many of the proposed algorithms strive to spread sensors to desired positions to obtain a stationary configuration such that the coverage is optimized, see e.g.,~\cite{JC-SM-TK-FB:04,AK-CG:07,AK-AS-CG:08,MS-DR-JJS:09,FB-RC-PF-:12,AC-MT-RC-LS-GP:15,MT-AC-RC-GP-LS:17,YC-SSK:20arxiv}. Some sensor placement problems such~\cite{MS-DR-JJS:09,AC-MT-RC-LS-GP:15,MT-AC-RC-GP-LS:17,YC-SSK:20arxiv} are context-aware, and include also a period of exploration and observation to increase the knowledge used to find the optimal residing position of the sensors.
In this paper, instead of aiming to achieve an improved stationary network configuration as the end result of the sensors' movement, our objective is to explore context-aware mobility strategies that dynamically reposition the mobile sensors to maximize their utilization and contribution over a mission horizon. Motivating applications include persistent monitoring to discover forest fires~\cite{yuan2015survey} or oil spillage in its early stages~\cite{henry2015wireless}, locating endangered animals in a large habitat~\cite{engler2004improved} and event detection in urban environments~\cite{TT-ECVB:90}. Specifically, we consider a persistent monitoring of a set of finite $\mathcal{V}$ inter-connected geographical nodes 
via a set of finite $\mathcal{A}$  mobile sensors/agents, where $|\mathcal{V}|>|\mathcal{A}|$. The mobile agents are confined to a set of pre-specified edges $\mathcal{E}\! 
\subset \!\mathcal{V} \times \mathcal{V}$, e.g., aerial or ground corridors, to traverse from one node to another, see Fig.~\ref{fig:city_partition}. 
Depending on their vehicle type, agents may have to take different edges to go from one node to another. Also, they may have different travel times along the same edge. 
We study dispatch policy that orchestrates the topological distribution of the mobile agents such that an optimized service for a global monitoring task is provided with a reasonable computational~cost. To quantify the service objective we assign to each node $v\in\VV$ the reward function, 
\begin{align}\label{eq::Ri}
    R_v(t)=\begin{cases}0,& t=\bar{t}_v,\\
    \psi_v(t-\bar{t}_v), &t>\bar{t}_v,\end{cases}
\end{align}
where $\psi_v(t)$ is a  nonnegative concave and increasing function of time and $\bar{t}_v$ is the latest time node $v$ is visited by an agent.  For example, in  data harvesting or health monitoring, $\psi_v(.)$ can be the weighted idle time of the node $v$ or in event detection, it can be the probability of at least one event taking place at inter-visit times. Optimal patrolling designs a dispatch policy (what sequence of nodes to visit at what times by which agents) to score the maximum collective reward for the team over the mission horizon. However, as we explain below, this problem is NP-hard. Our aim then is to design a suboptimal solution that has polynomial time complexity.

\begin{figure}
\centering
    {\includegraphics[width=.44\textwidth]{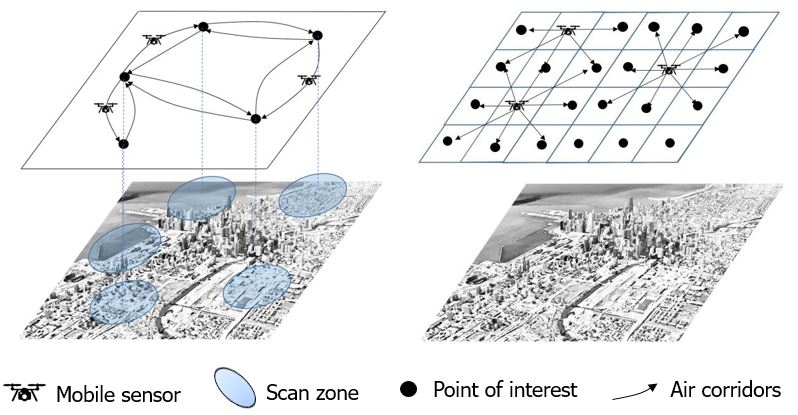}}\vspace{-0.1in}
    \caption{{\small Examples of a set of geographical nodes of interest and the edges between them. Finite number of nodes to monitor in a city can be restricted to some particular scanning zones (the picture on the left) or the cell partitioned map of the city (the picture on the right).}}
    \label{fig:city_partition}%
\end{figure}

\emph{Related work}: Dispatch policy design for patrolling/monitoring of geographical nodes can be divided into two categories: the edges to travel between the nodes are not specified (design in continuous edge space) or otherwise (design in discrete edge space). When there are no prespecified inter-node edges, the optimal patrolling policy design includes also finding the optimal inter-node trajectories that the agents should follow without violating their mobility limits. 
In some applications, however, the mobile agents are confined to travel through pre-specified known edges between the nodes. For example, in a smart city setting, regulations can restrict the admissible routes between the geographical nodes. 
In the dispatch policy design in discrete edge space, the complexity of finding the optimal policy for a single patrolling agent is the same as the complexity of solving the Traveling Salesman problem, where the computational complexity grows exponentially with the number of the nodes~\cite{RMK:72}. In case of multiple patrolling agents, the problem is even more complex, since each agent's policy design depends on the other agents' policy. This problem is formalized in earlier studies such as~\cite{machado2002multi,almeida2004recent}. Generally, when there are multiple edges to travel between every two nodes or when each node is connected to multiple other nodes, finding an optimal long term patrolling scheme is not tractable. Constraining the agents to travel through specific edges to traverse among the geographical nodes allows seeking optimal solutions for the problem. For example, when the connection topology between the geographical nodes is a path or a cyclic graph, optimal solutions for the problem are proposed in~\cite{chevaleyre2004theoretical,pasqualetti2012cooperative,yu2015persistent,donahue2016persistent}. To overcome the complexity issue on generic graphs,~\cite{ABA-SLS-SS:19} explores forming different cycles in the graph and assigning agents to these cycles to patrol the nodes periodically and seeks to minimize the time that a node stays un-visited. Alternatively,~\cite{farinelli2017distributed} proposes agents to move to the most rewarding neighboring node based on their current location.

\emph{Statement of contribution}: 
\begin{figure}
\centering
 {\includegraphics[width=.34\textwidth]{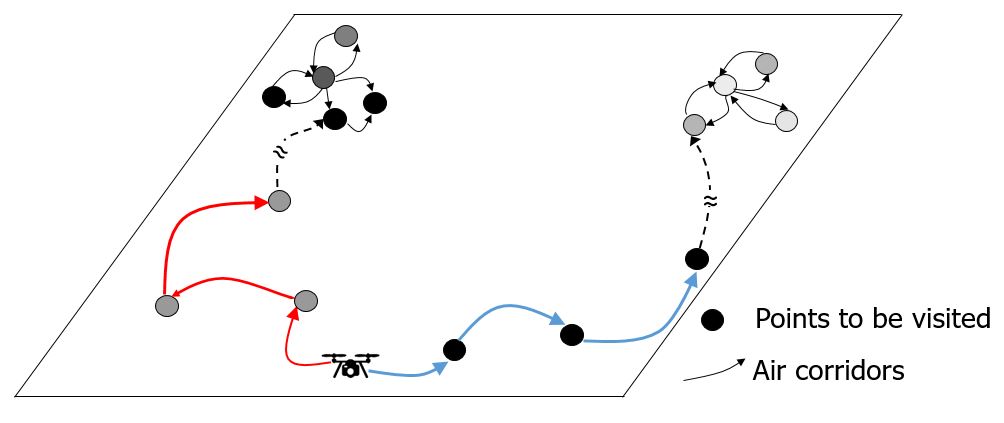}}\vspace{-0.1in}
    \caption{{\small An agent has two possible routes to take over the designated receding horizon. The nodes' color intensity shows their reward value.  The blue route offers a higher reward over the receding horizon but it puts the agent close to an area with a lower amount of reward, while the red route results in lower total reward over the receding horizon but puts the agent near an area with higher amount of reward.}}%
    \label{fig::nodal_importance}%
\end{figure}
In this paper, we propose a robust and suboptimal solution to the long term patrolling problem that we stated earlier. Instead of using the customary idle time, $\psi_v(t)=t$, as a reward function, which reduces the optimal dispatch policy design 
to the minimum latency problem~\cite{AB-PC-DC-BP-PR-MS:94}, we consider reward functions described by an increasing concave function. This allows modeling a wider class of patrolling problems such as patrolling for event detection. We let the utility function to be the sum of the rewards collected over the mission horizon by the mobile agents. We discuss that the design of optimal patrolling policy to maximize this utility over the mission horizon is an NP-hard problem. Specifically, we show that the complexity of finding the optimal policy increases exponentially with the mission horizon and number of agents. Next, we show that the utility function is a monotone increasing and submodular set function. To establish this result, we develop a set of auxiliary lemmas, presented in the appendix, based on the Karamata's inequality~\cite{ZK-DD-ML-IM:05}. 
Given the submodularity of the utility function, we propose a receding horizon sequential greedy algorithm to compute a suboptimal dispatch policy with a polynomial computation cost and guaranteed bound on optimality. The receding horizon nature of our solution induces robustness to uncertainties of the environment. Our next contribution is to add a new term to our utility function to compensate for the shortsightedness of the receding horizon approach, see Fig.~\ref{fig::nodal_importance}. When agents patrol a large set of inter-connected nodes, this added term becomes useful by giving them an intuition of the existing reward in the farther nodes. In recent years, submodular optimization has been widely used in sensor and actuator placement problems~\cite{AK-AS-CG:08,THS-FLC-JL:16,ZL-AC-PL-LB-DK-RP:18,AC-PL-BA-LB-RP:18,AK-CG:07,AC-LB-RP:14}. In comparison to the sensor/actuator placement problems, the challenge in our work is that the assigned policy per each mobile agent over the receding horizon is a dynamic scheduling problem rather than a static sensor placement. To deal with this challenge, we use the matroid constraint~\cite{MLF-GLN-LAW:78} approach to design our suboptimal submodular-based policy.  Finally, we discuss how our algorithm can be implemented in a decentralized manner. A simulation study demonstrates our results. Our notation is standard, though to avoid confusion, certain concepts and notation are defined as the need~arises. This paper extends our preliminary work~\cite{NR-SSK:19} in detailed technical treatment including all the proofs, introducing the notion of local importance to compensate for the shortsightedness of receding horizon approach, decentralized implementation of our algorithm, and a new simulation study. Also, we consider a more generalized case of reward functions.

\section{Problem Formulation}\label{sec::prob_def}
To formalize our objective, we first introduce our notations and state our standing assumptions. For any node $v\in\mathcal{V}$, $\mathcal{N}_v$ is a set consisting node $v$ and all the  neighboring nodes that are connected to node $v$ via an edge in $\mathcal{E}$. If there exists a path connecting node $v\in\mathcal{V}$ to node $w\in\mathcal{V}$, we let $\tau^i_{v,w}\in\real_{>0}$ be the \emph{shortest} travel time of agent $i\in\mathcal{A}$ from node $v$ to $w$. 


\begin{assump}\label{assm::scan}
 Upon arrival of any agent $i\in\mathcal{A}$ at any time $\bar{t}\in\real_{>0}$ at node $v\in\mathcal{V}$, the agent immediately scans the node and the reward $R_{v}(\bar{t})$ is scored for the patrolling team $\mathcal{A}$ and $\bar{t}_v$ of node $v$ in~\eqref{eq::Ri} is set to $\bar{t}$. If more than one agent arrives at node $v\in\mathcal{V}$ and scans it at the same time $\bar{t}$, the reward collected for the team is still $R_v(\bar{t})$. If an agent $i\in\mathcal{A}$ needs to linger over each node for  $\delta^i\in\real_{\geq 0}$ amount of time to complete its scan, during this time the agent cannot scan the node again to score a reward for the team.
\end{assump}

Let the tuple ${\mathsf{p}}=(\vectsf{V}_{\mathsf{p}},\vectsf{T}_{\mathsf{p}},\mathsf{a}_{\mathsf{p}})$ be a dispatch policy of agent $\mathsf{a}_{\mathsf{p}}\in\mathcal{A}$ over the given mission time horizon, where $\vectsf{V}_p$ and $\vectsf{T}_{\mathsf{p}}$ are the vectors that specify the nodes and the corresponding visit times assigned to agent $\mathsf{a}_{\mathsf{p}}$. Moreover, we let $\mathsf{n}_{\mathsf{p}}$ be the total number of nodes visited by agent $\mathsf{a}_{\mathsf{p}}$, i.e., $\mathsf{n}_{\mathsf{p}}=\text{dim}(\vectsf{V}_{\mathsf{p}})$. We refer to $\mathsf{n}_{\mathsf{p}}$ as the \emph{length} of the policy $\mathsf{p}$. 
We refer to $(\vectsf{V}_{\mathsf{p}}(l),\vectsf{T}_{\mathsf{p}}(l))$, $l\in \{1,2,\cdots,\mathsf{n}_{\mathsf{p}}\}$, as the $l^{\text{th}}$ \emph{step} of policy ${\mathsf{p}}$. Furthermore, for any agent $i\in\mathcal{A}$, we let $\mathcal{P}^i$ be the set of all the admissible policies $\mathsf{p}$ over the mission horizon such that $\mathsf{a}_{\mathsf{p}}=i$.
\begin{assump}\label{as::travel} 
For any policy $\mathsf{p}$, we have $\vectsf{V}_{\mathsf{p}}(l+1) \in \mathcal{N}_{\vectsf{V}_{\mathsf{p}}(l)}$,  for all $l\in \{1,2,\cdots,\mathsf{n}_{\mathsf{p}}-1\}$.
\end{assump}
We let  $\mathcal{P}=\bigcup_{i\in\mathcal{A}}^{}\mathcal{P}^i$. Then, given any
 $\bar{\mathcal{P}} \subset \mathcal{P}$, the utility function $\mathsf{R}:2^{\mathcal{P}} \to \realpositive$  is
 $\bar{\mathcal{P}} \subset \mathcal{P}$, the utility function $\mathsf{R}:2^{\mathcal{P}} \to \realpositive$  is
\begin{align}\label{eq::reward_global_op}
    \mathsf{R}(\bar{\mathcal{P}})=\sum\nolimits_{\forall \mathsf{p}\in\bar{\mathcal{P}}}\,\sum\nolimits_{l=1}^{\mathsf{n}_{\mathsf{p}}} R_{\vectsf{V}_{\mathsf{p}}(l)}(\vectsf{T}_{\mathsf{p}}(l)).
\end{align}
Given~\eqref{eq::reward_global_op}, the optimal policy to maximize the utility over a given mission horizon is given by 
\begin{subequations}\label{eq::rewardOpt_m}
\begin{align}
\label{eq::rewardOpt}
    \mathcal{P}^{\star}&=\underset{\bar{\mathcal{P}} \subset \mathcal{P}}{\text{argmax}}\,\mathsf{R}(\bar{\mathcal{P}}),\quad \text{s.t.} \\
\qquad \quad &|\bar{\mathcal{P}} \cap \mathcal{P}^i| \leq 1~~\qquad i\in \mathcal{A},\label{eq::rewardOpt-const}
\end{align}
\end{subequations}
where $|\,.\,|$ returns the cardinality of a set. The constraint condition~\eqref{eq::rewardOpt-const} is in the so-called \emph{partition matroid} form~\cite{MLF-GLN-LAW:78} and restricts the choice of the optimal solution to be a set that contains of at most one member from each disjoint sets $\mathcal{P}^i, \,\, i \in \mathcal{A}$.
A set value optimization problem of the form~\eqref{eq::rewardOpt_m} is known to be NP-hard~\cite{LL:83}. Lemma~\ref{lem::complexity} below, whose proof is given in the appendix, gives the cost of constructing the feasible set $\mathcal{P}$ and time complexity of solving optimization problem~\eqref{eq::rewardOpt_m}.

\begin{lem}[Time complexity of problem~\eqref{eq::rewardOpt}]\label{lem::complexity}
The cost of constructing the feasible set $\mathcal{P}$ of optimization problem~\eqref{eq::rewardOpt} is of order $O(\sum_{i\in \mathcal{A}}^{} \mathsf{D}^{\bar{\mathsf{n}}^i})$,
where 
$\mathsf{D}=\max_{v \in \mathcal{V}}(|\mathcal{N}_v|)$ and
$\bar{\mathsf{n}}^i=\max\{\mathsf{n}_{\mathsf{p}}\}_{\forall\mathsf{p}\in\mathcal{P}^i}$. Furthermore,
the time complexity of solving optimization problem~\eqref{eq::rewardOpt} is $O(\prod_{i\in \mathcal{A}}^{} \mathsf{D}^{\bar{\mathsf{n}}^i})$.
\end{lem}
If the system parameters, such as number of the mobile agents or the nodes, or the parameters of $\psi_v(.)$ of the reward function at any node $v$, change after the optimal policy design, the optimization problem~\eqref{eq::rewardOpt_m} should be solved again over the remainder of the mission horizon under the new conditions. Our objective in this paper is to construct a suboptimal solution to solve the persistent monitoring problem given by~\eqref{eq::rewardOpt_m} with polynomial time complexity. Moreover, we seek a solution that has intrinsic robustness to changes that can happen during the mission horizon.

We close this section by introducing some definitions and notations used subsequently.  For any set function $g:\,2^{\mathcal{Q}} \to \mathbb{R}$, we let 
$$\Delta_g (\mathsf{q}|\bar{\mathcal{Q}}) \!=\!g(\bar{\mathcal{Q}} \cup \mathsf{q} )-g(\bar{\mathcal{Q}} ),$$ for $\forall \bar{\mathcal{Q}} \in 2^\mathcal{Q}$ and $\forall \mathsf{q} \in \mathcal{Q}$, where $\Delta_g$ shows the increase in value of the set function $g$ going from set $\bar{\mathcal{Q}}$ to $\bar{\mathcal{Q}} \cup \mathsf{q}$. Recall that $g:\,2^{\mathcal{Q}} \to \mathbb{R}$ is \emph{submodular} if and only if for two sets $\mathcal{Q}_1$ and $\mathcal{Q}_2$ satisfying $\mathcal{Q}_1\subset \mathcal{Q}_2 \subset \mathcal{Q}$, and for $\mathsf{q} \not\in \mathcal{Q}_2$ we have~\cite{MLF-GLN-LAW:78}
$$\Delta_g (\mathsf{q}|\bar{\mathcal{Q}}_1) \geq \Delta_g (\mathsf{q}|\bar{\mathcal{Q}}_2).$$ Then submodularity is a property of set functions that shows diminishing reward as new members are being introduced to the system. We say $g:2^\mathcal{Q}\to\real$ is \emph{monotone increasing} if for all $\mathcal{Q}_1,\mathcal{Q}_2 \subset \mathcal{Q}$ we have 
    $\mathcal{Q}_1 \subset \mathcal{Q}_2$ if and only if~\cite{MLF-GLN-LAW:78} $$g(\mathcal{Q}_1)\leq g(\mathcal{Q}_2).$$ We denote a sequence of $m$ real numbers $(\mathfrak{t}_1,\cdots,\mathfrak{t}_m)$ by $(\mathfrak{t})_1^m$.
Given two increasing (resp. decreasing) sequences $(\mathfrak{t})_1^{n}$ and $(\mathfrak{v})_1^{m}$,  $(\mathfrak{t})_1^{n}\oplus(\mathfrak{v})_1^{m}$ is their concatenated increasing (resp. decreasing) sequence, i.e., for $(\mathfrak{u})_1^{n+m}=(\mathfrak{t})_1^{n}\oplus(\mathfrak{v})_1^{m}$, any $\mathfrak{u}_k$, $k\in\{1,\cdots,n+m\}$ is either in $(\mathfrak{t})_1^{n}$ or $(\mathfrak{v})_1^{m}$ or is in both. 
We assume that $(\mathfrak{u})_1^{n+m}$ preserves the relative labeling of $(\mathfrak{t})_1^{n}$ or $(\mathfrak{v})_1^{m}$, i.e., if $\mathfrak{t}_k$ and $\mathfrak{t}_{k+1}$, $k\in\{1,\cdots,n-1\}$ (resp. $\mathfrak{v}_k$ and $\mathfrak{v}_{k+1}$, $k\in\{1,\cdots,m-1\}$) correspond to $\mathfrak{u}_i$ and $\mathfrak{u}_j$ in $(\mathfrak{u})_1^{n+m}$, then $i<j$.

\section{Suboptimal policy design }\label{sec::submadul_policy}
According to Lemma~\ref{lem::complexity} the time complexity of finding an optimal patrolling policy in~\eqref{eq::rewardOpt} increases exponentially by the maximum length, $\bar{\mathsf{n}}^i$, of the admissible policies of any agent $i\in\mathcal{A}$ and also by the number of the exploring agents $M$. In light of this observation, to reduce the computational cost, we propose the following suboptimal policy design. Since the maximum policy length $\bar{\mathsf{n}}^i$ is proportional to the length of the mission horizon, we first propose to trade in optimality and divide the planning horizon into multiple shorter horizons so that the policy design can be carried out in a consecutive manner over these shorter horizons. Then, to reduce the optimality gap and also to induce robustness to the online changes that can occur during the mission time, we propose to implement this approach in a receding horizon fashion where we calculate the policy over a specified shorter horizon but execute only some of the initial steps of the policy, and then we repeat the process. However, a receding horizon approach suffers from what we refer to as 
\emph{shortsightedness}. That is, over large inter-connected geographical node sets, a receding horizon design is oblivious to the reward distribution of the nodes that are not in the feasible policy set in the planning horizon. Then, the optimal policy over the planning horizon can inadvertently steer the agents away from the distant nodes with a higher reward, see Fig.~\ref{fig::nodal_importance}. To compensate for this shortcoming, we introduce the notion of \emph{nodal importance} and augment the reward function~\eqref{eq::reward_global_op} over the design horizon with an additional term that given an admissible policy, provides a measure of how close an agent at the final step of the policy is to a cluster of geographical nodes with a high concentration of~reward. 

Let the augmented reward, whose exact form will be introduced below,  over the planning horizon be $\bar{\mathsf{R}}$. Then, the optimal policy design over each receding horizon is 
\begin{align}\label{eq::rewardOpt_nodal}
    \mathcal{P}^{\star}&=\underset{\bar{\mathcal{P}} \subset \mathcal{P}}{\text{argmax}}\,\bar{\mathsf{R}}(\bar{\mathcal{P}}),\quad \text{s.t.}\quad |\bar{\mathcal{P}} \cap \mathcal{P}^i| \leq 1, ~~i\in \mathcal{A}
\end{align}
where hereafter $\mathcal{P}=\bigcup_{i \in \mathcal{A}}^{}\mathcal{P}^i$ is the set of the union of the admissible policies of the agents $\mathcal{P}^i$, $i\in\mathcal{A}$, over the planning horizon. Hereafter, we let $\bar{\mathfrak{t}}^v_0$ be the last time node $v\in\mathcal{V}$ was visited before a planning horizon starts.

Next, to reduce the computational burden further, we propose to use Algorithm~\ref{alg:the_alg}, which is a  sequential greedy algorithm with a polynomial cost in terms of the number of the agents to obtain a suboptimal solution for~\eqref{eq::rewardOpt_nodal}.  In what follows, we show that since the objective function~\eqref{eq::rewardOpt_nodal} is a submodular set function, Algorithm~\ref{alg:the_alg} comes with a known optimality gap. We also show that with a proper inter-agent communication coordination Algorithm~\ref{alg:the_alg} can be implemented in a decentralized manner.

\begin{algorithm}[t]
\caption{Sequential Greedy  Algorithm}
\label{alg:the_alg}
{
\begin{algorithmic}[1]
\Procedure{\sf{SGOpt}}{$\mathcal{P}^i,\, i\in \mathcal{A}$}
\State $\mathbf{\text{Init:}} ~\bar{\mathcal{P}} \gets \emptyset, \,\, i\gets 0$, $\{\bar{\mathfrak{t}}^0_v\}_{v\in\mathcal{V}}$
\For{$i \in \mathcal{A}$}
\State $\mathsf{p}^{i\star}=\underset{\mathsf{p}\subset \mathcal{P}^i}{\argmax} \quad \Delta_{\bar{\mathsf{R}}} (\mathsf{p}|\bar{\mathcal{P}})$.
\State $\bar{\mathcal{P}} \gets\bar{\mathcal{P}}\cup \mathsf{p}^{i\star}$.
\EndFor
\State \textbf{Return} $\bar{\mathcal{P}}$.
\EndProcedure
\end{algorithmic}}
\end{algorithm}

For $v\in\mathcal{V}$, let $\mathcal{N}_v^r$ be the set consisted of node $v$ itself and its $r$-hope neighbors. This set can be computed using the Breadth-first search in time ${O}(|\mathcal{E}|+|\mathcal{V}|)$~\cite{THT-CEL-RLR-CS:09}. 
Here, $\tau^i_{w,v}$ can be computed via $A^\star$ algorithm in time $O(|\mathcal{E}|)$~\cite{FD-AB-MK-PB-MF-TF-LJ:14}. Then, for every node $v\in\VV$, we define the nodal importance with radius $r$ at time $\tau$ as $L(v,\tau,r)=\sum\nolimits_{w \in \mathcal{N}_v^r}{ }R_w(\tau)$. 
Next, given an agent $i\in\mathcal{A}$ that is at node $w\in\mathcal{V}$ at time $\hat{t}\in\real_{\geq0}$, we define the \emph{relative nodal importance} of a node $v\in\mathcal{V}$ with respect to agent $i$ as
$$
\mathsf{L}(v,w,\hat{t},i) ={L(v,\hat{t}+\tau^i_{w,v},r)}\big\slash{\tau^i_{w,v}}.$$
Then, $\mathsf{L}(v,\vectsf{V}_{\mathsf{p}}(\mathsf{n}_\mathsf{p}),\vectsf{T}_{\mathsf{p}}(\mathsf{n}_\mathsf{p}),\mathsf{a}_{\mathsf{p}})$ is a measure of the relative size of the awards concentration around any node $v\in\mathcal{V}$ that takes into account also the travel time of agent $\mathsf{a}_{\mathsf{p}}$ from the final step of policy ${\mathsf{p}}=(\vectsf{V}_{\mathsf{p}},\vectsf{T}_{\mathsf{p}},\mathsf{a}_{\mathsf{p}})\in\mathcal{P}$ to $v$. Let $\mathsf{L}(v,\mathsf{p})$ be the shorthand notation for  $\mathsf{L}(v,\vectsf{V}_{\mathsf{p}}(\mathsf{n}_\mathsf{p}),\vectsf{T}_{\mathsf{p}}(\mathsf{n}_\mathsf{p}),\mathsf{a}_{\mathsf{p}})$. To compensate for the shortsightedness of the receding horizon design, then we revise the utility function to 
\begin{align}\label{eq::reward_global_op_NI2}
    \bar{\mathsf{R}}(\bar{\mathcal{P}})= {\mathsf{R}}(\bar{\mathcal{P}})
    + \alpha \sum\nolimits_{\forall \mathsf{p}\in\bar{\mathcal{P}}}{}\underset{\forall v\in \bar{\mathcal{V}}}{\text{max }}\mathsf{L}(v,\mathsf{p}),~~~ \alpha\in\real_{\geq0}.
\end{align}
The weighting factor $\alpha\in\real_{\geq0}$ defines how much significance we want to assign to the distribution of the reward beyond the receding horizon. We should note that using a large $\alpha$ can gravitate the agents to move towards the nodes close to the anchor nodes, and make them oblivious to the rest of the nodes. For computational efficiency, instead of incorporating the relative nodal importance of all the nodes, which can be achieved by setting $\bar{\mathcal{V}}$ equal to $\mathcal{V}$, we propose to use only $\bar{\mathcal{V}}$ subset of the nodes. We refer to nodes in $\bar{\mathcal{V}}$ as \emph{anchor nodes}. The anchor nodes can be selected to be the nodes with higher reward return or to be a set of nodes that are scattered uniformly on the graph. It is interesting to note that the relative nodal importance term in~\eqref{eq::reward_global_op_NI2} is a reminiscent of terminal cost used in the model predictive control (MPC). In MPC, terminal cost that is used to achieve an infinite horizon control with closed-loop stability guarantees~\cite{ECG-DMP-MM:89} in some way also compensates for the shortsightedness of the design over finite planning horizon. Next, we show that the reward function~\eqref{eq::reward_global_op_NI2} is submodular over any given feasible policy set $\mathcal{P}$ in every planning horizon.  

\begin{thm}[Submodularity of the reward function~\eqref{eq::reward_global_op_NI2}]\label{thm::submod}
For any weighting factor $\alpha\in\real_{\geq0}$, the reward function $\bar{\mathsf{R}}:2^{\mathcal{P}} \to \realpositive$ in~\eqref{eq::reward_global_op_NI2}
is a monotone increasing and submodular set function over $\mathcal{P}$.
\end{thm}

\begin{proof}
Let  $\mathtt{c}(v,{\mathcal{Q}}):\mathcal{V} \times 2^{\mathcal{Q}}\to\mathbb{Z}_{>0}$ be the total number of visits to the geographical node $v$, and ${\mathcal{I}}_{{\mathcal{Q}}} \subset \mathcal{V}$ be the set of the nodes that are visited when a policy set ${\mathcal{Q}} \subset\mathcal{P}$ is implemented. Furthermore, let the increasing sequence 
$$(\mathfrak{t}^v({\mathcal{Q}}))_1^{c(v,{\mathcal{Q}})} = ( \mathfrak{t}^v_1({\mathcal{Q}}) , \mathfrak{t}^v_2({\mathcal{Q}}) , \cdots ,\mathfrak{t}^v_{c(v,{\mathcal{Q}})}({\mathcal{Q}}))$$
be the sequence of time that node $v \in {\mathcal{I}}_{{\mathcal{Q}}}$ was visited when agents implement ${\mathcal{Q}}$. Now consider the reward function $\bar{\mathsf{R}}$ in \eqref{eq::reward_global_op_NI2}. Then, the first summand of $\bar{\mathsf{R}}$ expands as  $${\mathsf{R}}(\bar{\mathcal{P}})= \sum\nolimits_{v\in {\mathcal{I}}_{\bar{\mathcal{P}}}}{}\big(\sum\nolimits_{j=1}^{c(v,\bar{\mathcal{P}})} \psi_v(\Delta \mathfrak{t}_j^v(\bar{\mathcal{P}}))\big)
    $$, where $\Delta \mathfrak{t}_j^v(\bar{\mathcal{P}})=\mathfrak{t}^v_j(\bar{\mathcal{P}})-\mathfrak{t}^v_{j-1}(\bar{\mathcal{P}})$ is the time between two consecutive visits of node $v$, and $\mathfrak{t}^v_0(\bar{\mathcal{P}})=\bar{\mathfrak{t}}^v_0$. 
Next, consider the monitoring policy sets $\mathcal{Q}_1,\,\mathcal{Q}_2$ and monitoring policy $\mathsf{q}$ with $\mathcal{Q}_1 \subset \mathcal{Q}_2 \subset \mathcal{P}$, $\mathsf{q} \in \mathcal{P}$, $\mathsf{q}\not\in \mathcal{Q}_1,$ and $\mathsf{q}\not\in \mathcal{Q}_2$.
Because $(\mathfrak{t}^v(\mathcal{Q}_1))_1^{c(v,\mathcal{Q}_1)}$ is a sub-sequence of $(\mathfrak{t}^v(\mathcal{Q}_2))_1^{c(v,\mathcal{Q}_2)}$, using Lemma~\ref{lem::auxMain1} and the fact that $\psi(.)_v$ is a normalized increasing concave function, we conclude that $$ \!\sum\nolimits_{j=1}^{c(v,\mathcal{Q}_1 \cup \mathsf{q})}\!\! \! \psi_v(\Delta (\mathfrak{t}_j^v(\mathcal{Q}_2\cup\mathsf{q}))) \!- \!\! \sum\nolimits_{j=1}^{c(v,\mathcal{Q}_1)}\!\!\!\!\psi_l(\Delta (\mathfrak{t}_j^v(\mathcal{Q}_2)))
   \!\geq 0$$ for $\forall v \in {\mathcal{I}}_{\bar{\mathcal{P}}}$.
Therefore, $\Delta_{\mathsf{R}}(p|\mathcal{Q}_1)\geq 0$ which shows that $\mathsf{R}(\bar{\mathcal{P}})$ is a monotone increasing set function. Furthermore, using Lemma~\ref{lem::auxMain2} we can write 
\begin{align*}
   \SUM{j=1}{c(v,\mathcal{Q}_2 \cup \mathsf{q})} \psi_v(\Delta (\mathfrak{t}_j^v(\mathcal{Q}_2\cup\mathsf{q}))) &-  \SUM{j=1}{c(v,\mathcal{Q}_2)}\psi_v(\Delta (\mathfrak{t}_j^v(\mathcal{Q}_2)))\\
   &\leq \\
   \SUM{j=1}{c(v,\mathcal{Q}_1 \cup \mathsf{q})} \psi_v(\Delta (\mathfrak{t}_j^v(\mathcal{Q}_1\cup\mathsf{q}))) &-  \SUM{j=1}{c(v,\mathcal{Q}_1)}\psi_v(\Delta (\mathfrak{t}_j^v(\mathcal{Q}_1))).
\end{align*}
Hence, $$\Delta_{\mathsf{R}}( \mathsf{q}|\mathcal{Q}_1)\geq\Delta_{\mathsf{R}}( \mathsf{q}|\mathcal{Q}_2)$$ which shows that ${\mathsf{R}}(\bar{\mathcal{P}})$ is a submodular set function. Then, since the second summand of $\bar{\mathsf{R}}$,  $\sum\nolimits_{\forall \mathsf{p}\in\bar{\mathcal{P}}}{}\underset{\forall l\in \bar{\mathcal{V}}}{\text{max }}\mathsf{L}(l,\mathsf{p})$, is trivially positive and modular, the proof is concluded.
\end{proof}

Due to Theorem~\ref{thm::submod}, the suboptimal dispatch policy of Algorithm~\ref{alg:the_alg}, which has a polynomial computational complexity, has the following well-defined optimality gap.
\begin{thm}[Optimality gap of Algorithm~\ref{alg:the_alg}]\label{thm::opt_gap_Alg1}
Let $\mathcal{P}^{\star}$ be an optimal solution of~\eqref{eq::rewardOpt_nodal} and  $\bar{\mathcal{P}}$ be the output of Algorithm~\ref{alg:the_alg}. Then, $\bar{\mathsf{R}}(\bar{\mathcal{P}}) \geq \frac{1}{2}\bar{\mathsf{R}}(\mathcal{P}^\star)$.
\end{thm}

\begin{proof}
Since the objective function of~\eqref{eq::rewardOpt_nodal}  is monotone increasing and submodular over $\mathcal{P}$, the proof follows by invoking~\cite[Theorem 5.1]{MLF-GLN-LAW:78}.
\end{proof}

\subsection{Comments on decentralized implementations of Algorithm~\ref{alg:the_alg}}
To implement Algorithm~\ref{alg:the_alg},  given the current position of each agent and $\{\bar{\mathfrak{t}}^0_v\}_{v\in\mathcal{V}}$ at the beginning of each planning horizon, the admissible set of policies $\mathcal{P}^i$ for each agent $i\in\mathcal{A}$ should be calculated. 

Let every agent know $\{\psi_v(t)\}_{v\in\mathcal{V}}$.
A straightforward decentralized implement of Algorithm~\ref{alg:the_alg} then is a multi-centralized solution. In this solution, agents transmit the feasible policy sets across the entire network until each agent knows the whole policy set $\mathcal{P}^i,\,\, \forall i \in \mathcal{A}$ (flooding approach). 
Then, each agent acts as a central node and runs a copy of Algorithm~\ref{alg:the_alg} locally.
Although reasonable for small-size networks, the communication and storage costs of this approach scale poorly with the network size. The sequential structure of Algorithm~\ref{alg:the_alg} however, offers an opportunity for a communicationally and computationally more efficient decentralized implementations, as described in steps 1 to 9 of Algorithm~\ref{alg:the_alg_decen}. Step 10 of Algorithm~\ref{alg:the_alg_decen} is included for receding horizon implementation purpose, where the execution plan  can be for example one or all of the agents visit at least one node. To implement Algorithm~\ref{alg:the_alg_decen}, we assume that the agents $\mathcal{A}$ can form a bidirectional connected communication graph $\mathcal{G}^a=(\mathcal{A},\mathcal{E}^a)$, i.e., there is a path from every agent to every other agent on $\mathcal{G}^a$. Then, there always exists a route $\mathtt{SEQ} = \mathtt{s}_1 \to \cdots\to\mathtt{s}_i\to\cdots \to \mathtt{s}_K$,  $\mathtt{s}_k \in \mathcal{A}$, $k\in \{1,\cdots,K\}$, $K\!\geq\! M$, that visits all the agents (not necessarily only one time), see Fig.~\ref{Fig::decentralized}(a). The agents follow $\mathtt{SEQ}$ to share their information while implementing Algorithm~\ref{alg:the_alg_decen}.  
\begin{algorithm}[t]
\caption{Decentralized Implementation of Sequential Greedy Algorithm}
\label{alg:the_alg_decen}
{
\begin{algorithmic}[1]
\State $\mathbf{\text{Init:}} ~\bar{\mathcal{P}} \gets \emptyset, \,\, i\gets 1$, $\{\bar{\mathfrak{t}}_v^0\}_{v\in\mathcal{V}}$
\While{$i\leq K$}
\If {$\mathtt{s}_i$ \text{is being called for the first time}}
\State agent $\mathtt{s}_i$ computes $\mathsf{p}^{\mathtt{s}_i\star}=\underset{\mathsf{p}\subset \mathcal{P}^{\mathtt{s}_i}}{\argmax} \quad \Delta_{\bar{\mathsf{R}}} (\mathsf{p}|\bar{\mathcal{P}})$.
\State $\bar{\mathcal{P}} \gets\bar{\mathcal{P}}\cup \mathsf{p}^{\mathtt{s}_i\star}$.
  \EndIf
\State agent $\mathtt{s}_i$ \text{pass} $\bar{\mathcal{P}}$ to $\mathtt{s}_{i+1}$.
\State $i \gets i+1$.
\EndWhile\\
agent $\mathtt{s}_K$ based on the execution plan of the receding horizon operation updates $\{\bar{\mathfrak{t}}_v^0\}_{v\in\mathcal{V}}$ and communicates it to the team 
\end{algorithmic}
}
\end{algorithm} 
The communication cost to execute Algorithm~\ref{alg:the_alg_decen} can be optimized by picking $\mathtt{SEQ}$ to be the shortest path~\cite{ELL-JKL-AHGRK-DBS:85} that visits all the agents over graph $\mathcal{G}^a$. If $\mathcal{G}^a$ has a Hamiltonian path, the optimal choice for $\mathtt{SEQ}$ is a Hamiltonian path. Recall that a Hamiltonian path is a path that visits every agent on $\mathcal{G}^a$ only once~\cite{FR:74}. When, there is a $\mathtt{SEQ}$ that visits every agent on $\mathcal{G}^a$, the directed information graph $\mathcal{G}^I=(\mathcal{A},\mathcal{E}^I)$ of Algorithm~\ref{alg:the_alg_decen}, which shows the information access of each agent while implementing Algorithm~\ref{alg:the_alg_decen}, is full, see Fig.~\ref{Fig::decentralized}. That is, each agent in $\mathtt{SEQ}$ is aware of the previous agents' decision. Therefore, the solution obtained by Algorithm~\ref{alg:the_alg_decen} is an exact sequential greedy algorithm and its optimality gap is $1/2$. We recall that the labeling order of the mobile agents does not have an effect on the optimality gap guaranteed by Theorem~\ref{thm::opt_gap_Alg1}~\cite{BG-SLS:18}. If an agent $i\in\mathcal{A}$ appears repeatedly in $\mathtt{SEQ}$ (e.g., the blue agent in Fig.~\ref{Fig::decentralized}), with a slight increase in  computation cost, we can modify Algorithm~\ref{alg:the_alg_decen} to allow agent $i$ to redesign and improve its sub-optimal policy $\mathsf{p}^{i\star}$ by re-executing step 4 of Algorithm~\ref{alg:the_alg_decen}.

Another form of decentralized implementation of Algorithm~\ref{alg:the_alg}, which may be more relevant in urban environments, is through a client-server framework implemented over a cloud. In this framework, agents (clients) connect to shared memory on a cloud (server) to download or upload information or use the cloud's computing power asynchronously. 
Let $\{\mathcal{T}^i\}, \, i\in\mathcal{A}$, be the set of disjoint time slots that is allotted respectively to agents $\mathcal{A}$, see Fig.~\ref{fig:server_client}. To implement Algorithm~\ref{alg:the_alg}, agent $i\in\mathcal{A}$ connects to the server at the beginning of $\mathcal{T}^i$ to check out $\bar{\mathcal{P}}$ and $\{\bar{\mathfrak{t}}_v^0\}_{v\in\mathcal{V}}$. Then, it completes steps $4$ and $5$ of Algorithm~\ref{alg:the_alg}, and checks in the updated $\bar{\mathcal{P}}$ to the server before $\mathcal{T}^i$ elapses fully. The last agent based on the execution plan of the receding horizon operation updates $\{\bar{\mathfrak{t}}_v^0\}_{v\in\mathcal{V}}$ and checks it in the cloud memory for next receding horizon planning. Since the time slots assigned to the agents do not overlap, agent $i$ has access to policy $\mathsf{p}^{k\star}$ of all agents $k$ which has already communicated to the cloud. Thus, the information graph $\mathcal{G}^I$ is full, and the optimality gap of $1/2$ holds.

If there is a message dropout while executing Algorithm~\ref{alg:the_alg_decen} or in the decentralized server-client based operation an agent $j$ takes a longer time than $\mathcal{T}^j$ to complete and check-in $\bar{\mathcal{P}}$ to the cloud, the information graph becomes incomplete, see for example Fig.~\ref{fig:server_client}. Then, the corresponding decentralized implementation deviates from the exact sequential greedy Algorithm~\ref{alg:the_alg_decen}. For such cases,~\cite{BG-SLS:18} shows that the optimality gap instead of $1/2$ becomes $\frac{1}{M-\omega(\mathcal{G}^I)+2}$, where $\omega(\mathcal{G}^I)$ is the clique number of $\mathcal{G}^I$~\cite{BG-SLS:18}. Recall that the clique number of a graph is equal to the number of the nodes in the largest sub-graph such that adding an edge will cause a cycle \cite{JAB-USRM:76}.

\begin{figure}[t]
    \centering
    {
     \begin{tikzpicture}[auto,thick,scale=0.4, every node/.style={scale=0.4},node distance=2.5cm]
\tikzset{edge/.style = {->,> = latex'}}
el/.style = {inner sep=2pt, align=left, sloped},
             
\node (1)at (-4.5,0)  [inner sep=0pt] {\includegraphics[width=.045\textwidth]{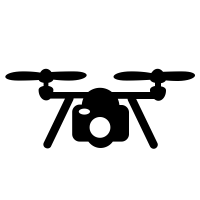}};

\node (2) [right of=1,inner sep=0pt] {\includegraphics[width=.045\textwidth]{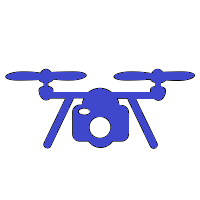}};

\node (3) [below left of=2,inner sep=0pt] {\includegraphics[width=.045\textwidth]{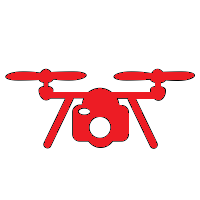}};

\node (4) [below right of=2,inner sep=0pt] {\includegraphics[width=.045\textwidth]{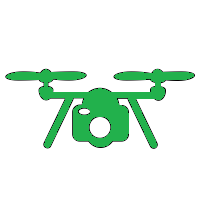}};

\node (5) [right of=2,inner sep=0pt] {\includegraphics[width=.045\textwidth]{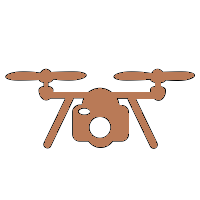}};

\draw[latex'-latex']  (1)--(2);
\draw[latex'-latex'] (3)-- (4) ;
\draw[latex'-latex']  (4)-- (2) ;
\draw[latex'-latex'] (2)-- (3) ;
\draw[latex'-latex']  (2)-- (5) ;

\path[edge,dashed,bend left, bend angle = 5 , color = red] (1) edge node[above]{$1$} (2);
\path[edge,dashed,bend right, bend angle = 5, color = red] (2) edge node[left]{$2$} (3);
\path[edge,dashed,bend right, bend angle = 5, color = red] (3) edge node[below]{$3$} (4);

\path[edge,dashed,bend left, bend angle = 5, color = red] (2) edge node[above]{$5$} (5);

\path[edge,dashed,bend left, bend angle = 5, color = red] (4) edge node[right]{$4$} (2);

\node (6)at (4.5,0) [inner sep=0pt] {\includegraphics[width=.045\textwidth]{drone1.png}};

\node (7) [right=0.5cm of 6,inner sep=0pt] {\includegraphics[width=.045\textwidth]{drone2.png}};

\node (8) [right=0.5cm of 7,inner sep=0pt] {\includegraphics[width=.045\textwidth]{drone3.png}};

\node (9) [right=0.5cm of 8,inner sep=0pt] {\includegraphics[width=.045\textwidth]{drone4.png}};

\node (10) [right=0.5cm of 9,inner sep=0pt] {\includegraphics[width=.045\textwidth]{drone5.png}};

\draw[edge]  (6)to (7);
\draw[edge]  (7)to (8) ;
\draw[edge]  (8)to (9) ;
\draw[edge]  (9)to (10) ;

\path[edge,bend left, bend angle = 5] (6) edge (8);
\path[edge,bend left, bend angle = 5] (6) edge (9);
\path[edge,bend left, bend angle = 5] (6) edge (10);

\path[edge,bend right, bend angle = 5] (7) edge (9);
\path[edge,bend right, bend angle = 5] (7) edge (10);

\path[edge,bend left, bend angle = 5] (8) edge (10);

\end{tikzpicture}
}
\caption{{\small
  The plot on the left shows the  bi-directional communication graph $\mathcal{G}^a$ in black along with an example $\mathtt{SEQ}$ path in red. The plot on the right shows the complete information sharing graph $\mathcal{G}^I$ if agents follow $\mathtt{SEQ}$ while implementing Algorithm~\ref{alg:the_alg_decen}. Arrow going from agent $i$ to agent $j$ means that agent $j$ receives agent $i$'s information.}
 }
\label{Fig::decentralized}
\end{figure}

\begin{figure}[t]
    \centering
    {
    {\includegraphics[trim=0 3pt 0 0,clip,width=.18\textwidth]{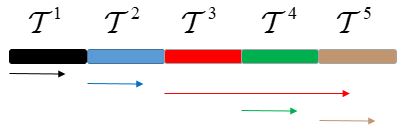}}\quad
    {
     \begin{tikzpicture}[auto,thick,scale=0.40, every node/.style={scale=0.40},node distance=2.5cm]
\tikzset{edge/.style = {->,> = latex'}}
el/.style = {inner sep=2pt, align=left, sloped},

\node (6)at (-0.2, -4.5) [inner sep=0pt] {\includegraphics[width=.045\textwidth]{drone1.png}};

\node (7) [right=0.5cm of 6,inner sep=0pt] {\includegraphics[width=.045\textwidth]{drone2.png}};

\node (8) [right=0.5cm of 7,inner sep=0pt] {\includegraphics[width=.045\textwidth]{drone3.png}};

\node (9) [right=0.5cm of 8,inner sep=0pt] {\includegraphics[width=.045\textwidth]{drone4.png}};

\node (10) [right=0.5cm of 9,inner sep=0pt] {\includegraphics[width=.045\textwidth]{drone5.png}};

\draw[edge]  (6)to (7);
\draw[edge]  (7)to (8) ;
\draw[edge]  (9)to (10) ;

\path[edge,bend left, bend angle = 5] (6) edge (8);
\path[edge,bend left, bend angle = 5] (6) edge (9);
\path[edge,bend left, bend angle = 5] (6) edge (10);

\path[edge,bend right, bend angle = 5] (7) edge (9);
\path[edge,bend right, bend angle = 5] (7) edge (10);

\end{tikzpicture}
}}
\caption{{\small
  $\{\mathcal{T}^i\}_{i \in \mathcal{A}}, \, \mathcal{A}=\{1,2,3,4,5\}$ are the time slots allotted to each agent to connect to the cloud. The arrows show the time each agent took to do their calculations for an example scenario. Here, the associated information graph $\mathcal{G}^I$ is as the incomplete graph on the right with clique number of 3.}
 }
\label{fig:server_client}
\end{figure}

\section{Numerical Example}
We consider persistent monitoring using $3$ agents for event detection over an area that is divided into $20$ by $20$ grid map as shown in Fig.~\ref{fig:ex1}(a). The geographical nodes of interest $\mathcal{V}$ are the center of the cells in Fig.~\ref{fig:ex1}(a). The agents can travel from a cell to the neighboring cells in the right, left, bottom, and top. The agents are homogeneous and the travel time between any neighboring nodes for all the agents are identical and equal to $1$ second. The agents start their patrolling task from the nodes where they are depicted in Fig.~\ref{fig:ex1}(a). We model the event occurrence in each geographical node as a Poisson process and define our reward function at each node $v\in\mathcal{V}$ as~\eqref{eq::Ri} with $\psi_v(t)=1-\text{e}^{\lambda_v t}$ where $\lambda_v\in\real_{>0}$ is the arrival rate of the event; for more details see~\cite{NR-SSK:19}. Fig.~\ref{fig:ex1}(a) shows the reward value of the nodes at $t=120$ seconds when there is no monitoring. The color intensity of the cells in Fig.~\ref{fig:ex1}(a) is proportional to $\lambda_v$; the higher $\lambda_v$, the darker the color of node $v$. The region enclosed by the blue rectangle initially has a low reward but after $100$ seconds its reward value is increased to a higher value by changing $\lambda_v$ of the corresponding cells.  An animated depiction of the change in the reward map because of different dispatch policies we discuss below is available in~\cite{NR-SSK:19-youtube}. We compare the performance of Algorithm~\ref{alg:the_alg}, implemented in a receding horizon fashion, and a conventional greedy algorithm where each agent always moves to the neighboring node that has the instantaneous highest reward value. In implementing Algorithm~\ref{alg:the_alg} in a receding horizon fashion, we assume that the planning horizon is $4$ seconds and the execution horizon is $1$ second. We consider both the case of including ($\alpha=0.1$) and excluding ($\alpha=0$) the nodal importance measure in the reward function~\eqref{eq::reward_global_op_NI2}. Fig.~\ref{fig:ex1}(b) shows that the traditional greedy cell selection performs poorly compared to the other two planning algorithms. The reason is that the three agents' decision becomes the same after a while, i.e., they start choosing the same cell after a while and moving together, therefore all three agents act as if one agent is patrolling (recall Assumption~\ref{assm::scan}). The performance of Algorithm~\ref{alg:the_alg} is better than a standard greedy cell selection because the effect of agent $i$'s patrolling policy is taken into account when agent $i+1$ is designed. Therefore, the chances that all three agents go to the same cell together and move together is narrow. Furthermore, we can note that implementing Algorithm~\ref{alg:the_alg} by considering the effect of nodal importance delivers a better outcome. The reason is that in the case that there is no nodal importance, the agents are drawn to the region of high importance near them and stay there as Fig.~\ref{fig:ex1}(c) shows. However, there are other important regions with higher values that are farther away, especially the area on the left top corner which is separated by a low rate stripe from where agents start. Incorporating nodal importance, as Fig.~\ref{fig:ex1}(d) shows steers the agents to the regions with a higher rate of reward that are beyond the receding horizon's sight.

\begin{figure}
\centering
\subfloat[Reward map ]{
    {\includegraphics[trim=20pt 10pt 10pt 5pt ,clip,width=.20\textwidth]{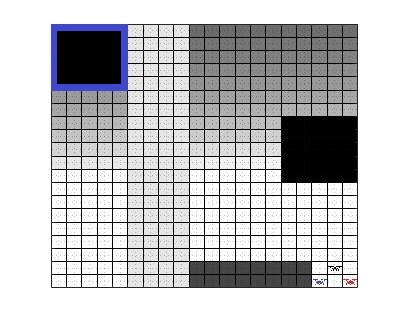}}}~
\subfloat[The collected reward]{    
    {\includegraphics[trim=2pt  5pt 10pt 20pt, clip,width=.22\textwidth]{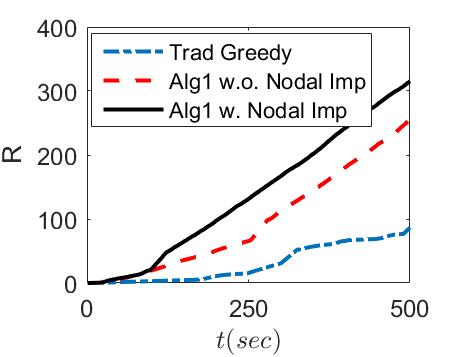}}}\\\vspace{-0.1in}
\subfloat[Agents' path when they follow Algorithm~\ref{alg:the_alg} and use $\alpha=0$ over ${[}0,150{]}$ seconds]{
    {\includegraphics[trim=30pt 30pt 30pt 20pt ,clip,width=.215\textwidth]{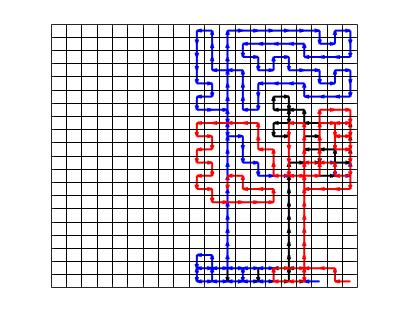}}}~~
\subfloat[Agents' path when they follow Algorithm~\ref{alg:the_alg} and use $\alpha=0.1$ over ${[}0,150{]}$ seconds]{    
    {\includegraphics[trim=30pt 30pt 30pt 20pt ,clip,width=.215\textwidth]{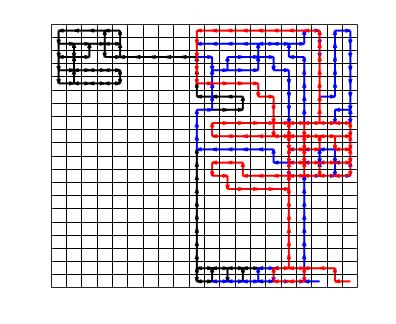}}}
    \caption{{\small Three agents patrol a field, divided into $20$ by $20$ cells.}}%
    \label{fig:ex1}%
\end{figure}

We presented a multi-agent dispatch policy design for persistent monitoring of a set of finite inter-connected geographical nodes. Our design relied on assigning an increasing and concave reward function of time to each node that reset to zero after a visit by an agent. We defined our design utility function as the sum of the rewards scored for the team when agents visit the geographical nodes. By showing that the utility function is a monotone increasing and submodular set function, we laid the ground to propose a suboptimal solution with a known optimality gap for our dispatch policy design, which was NP-hard. To induce robustness to the changes in the problem parameters, we proposed our suboptimal solution in a receding horizon setting. Next, to compensate for the shortsightedness of the receding horizon approach, we added a new term, called the relative nodal importance, to the utility function as a measure to incorporate a notion of the importance of the regions beyond the feasible solution set of the receding horizon optimization problem. Our numerical example demonstrated the benefit of introducing this term. Lastly, we discussed how our suboptimal solution can be implemented in a decentralized manner. Our future work is to investigate decentralized algorithms that allow agents to communicate synchronously with each other in order to have a consensus on a policy with a known optimality~gap.

\section*{Appendix}
\setcounter{equation}{0}
\renewcommand{\theequation}{A.\arabic{equation}}
\renewcommand{\thethm}{A.\arabic{thm}}
\renewcommand{\thelem}{A.\arabic{lem}}
\textbf{[Proof of Lemma~\ref{lem::complexity}]}
The time complexity of constructing the admissible policy set $\mathcal{P}^i$ is of order of the number of possible paths that agent $i\in\mathcal{A}$ can traverse over the mission horizon while respecting Assumption~\ref{as::travel}, which is of order $\mathsf{D}^{\bar{\mathsf{n}}^i}$. Thus, the time complexity of constructing the feasible set  $\mathcal{P}\!=\!\bigcup_{i \in \mathcal{A}}^{}\mathcal{P}^i$ is $O(\sum_{i \in \mathcal{A}}^{} \mathsf{D}^{\bar{\mathsf{n}}^i})$. Next, let $\bar{\mathcal{P}}$ be any subset of $\mathcal{P}$ that satisfies constraint~\eqref{eq::rewardOpt-const}. Due to  Assumption~\ref{assm::scan}, the reward scored by implementing policy ${\mathsf{p}}\!=\!(\vectsf{V}_{\mathsf{p}},\vectsf{T}_{\mathsf{p}},\mathsf{a}_{\mathsf{p}})\in\bar{\mathcal{P}}$ cannot be calculated independent from the all the other policies in $\bar{\mathcal{P}}\backslash\{\mathsf{p}\}$.
Hence, to solve optimization problem~\eqref{eq::rewardOpt}, we need to evaluate all the possible policy sets $\bar{\mathcal{P}}$ satisfying the constraint~\eqref{eq::rewardOpt-const}. Since $\bar{\mathcal{P}}$ can have at most one policy from the policy set $\mathcal{P}^i$ of $i\! \in\! \mathcal{A}$ and $\mathcal{P}^i$ has $O(\sum_{i \in \mathcal{A}}^{} \mathsf{D}^{\bar{\mathsf{n}}^i})$ members, then $O(\prod_{i=1}^{M} \mathsf{D}^{\bar{\mathsf{n}}^i})$ different possibilities of $\bar{\mathcal{P}}$ exist which determines the time complexity of solving~\eqref{eq::rewardOpt}. 
\boxend

We develop the auxiliary results below to use in the proof of Theorem~\ref{thm::submod}.  These results show some of the properties of the sum of evaluation of a concave and increasing function over increasing sequences and their concatenation. 
The decreasing sequence $(\delta \mathfrak{t})_1^n$ \emph{majorizes} the decreasing sequence $(\delta \mathfrak{v})_1^n$, if $$\delta \mathfrak{t}_1 \!\geq\! \delta \mathfrak{t}_2 \!\geq \!\cdots \!\geq\! \delta \mathfrak{t}_n,$$ $$\delta \mathfrak{v}_1 \!\geq\! \delta \mathfrak{v}_2 \!\geq \!\cdots \!\geq\! \delta \mathfrak{v}_n,$$ $$\delta \mathfrak{t}_1+\cdots+\delta \mathfrak{t}_i \geq \delta \mathfrak{v}_1+\cdots+\delta \mathfrak{v}_i,\,\,\, \forall i\! \in\! \{1,\cdots,n-1 \}$$ and $$\delta \mathfrak{t}_1+\cdots+\delta \mathfrak{t}_n = \delta \mathfrak{v}_1+\cdots+\delta \mathfrak{v}_n$$~hold.

\begin{lem}\label{lem::krm}
Let $f:\real \to \real$ be a concave and increasing function with $f(0)=0$. If sequences $(\delta \mathfrak{t})_1^n$ and $(\delta \mathfrak{v})_1^m$ with $n\leq m$ satisfy $\delta \mathfrak{t}_1+\cdots+\delta \mathfrak{t}_i \geq \delta \mathfrak{v}_1+\cdots+\delta \mathfrak{v}_i, \quad \forall i \in \{1,\cdots,n-1 \}$ and $\delta \mathfrak{t}_1+\cdots+\delta \mathfrak{t}_n = \delta \mathfrak{v}_1+\cdots+\delta \mathfrak{v}_m$ then $$f(\delta \mathfrak{t}_1)+\cdots+f(\delta \mathfrak{t}_n) \leq f(\delta \mathfrak{v}_1)+\cdots+f(\delta \mathfrak{v}_m)$$ holds.
\end{lem}

\begin{proof}
We note that the sequence $(\delta \mathfrak{u})_1^m$ defined as $\delta \mathfrak{u}_i = \delta \mathfrak{u}_i$ for $i \in \{1,\cdots,n\}$ and $\delta \mathfrak{u}_i = 0$ for  $i \in \{n+1,\cdots,m\}$ majorizes any sequence $(\delta \mathfrak{v})_1^m$ defined in the lemma statement. Then, since $f(0)=0$, the proof follows from the Karamata's inequality \cite{ZK-DD-ML-IM:05}. 
\end{proof}

\begin{cor}\label{lm::aux1}
Let $f:\realnonnegative \to \realnonnegative$ be a monotone increasing and concave function. Then for any $a,b,c,d\in\real_{\geq0}$ such that $0\leq a \leq c$ and $0\leq b \leq d$, then  $$f(c)+f(d)-f(c+d) \leq f(a)+f(b)-f(a+b)$$ holds.
\end{cor}

\begin{proof}
The assumption is that $a\leq c$ and $b\leq d$. Therefor, we have $a+b\leq c+d$. By taking $a,b,c+d$ and $c,d,a+b$ as $\delta \mathfrak{t}$'s $\delta \mathfrak{v}$'s respectively. There will be two possible cases for $\delta \mathfrak{t}$'s as 
\begin{align*}
&(A1):\, \,\, \, \delta \mathfrak{t}_1= c+d, \, \, \, \delta \mathfrak{t}_2= a,  \, \, \, \delta \mathfrak{t}_3= b,\\
&(A2):\, \,\, \, \delta \mathfrak{t}_1= c+d, \, \, \, \delta \mathfrak{t}_2= b,  \, \, \, \delta \mathfrak{t}_3= a,
\end{align*}
and there will be six possible cases for $\delta \mathfrak{v}$'s as 
\begin{align*}
&(B1):\, \,\, \, \delta \mathfrak{v}_1= a+b, \, \, \, \delta \mathfrak{v}_2= d,  \, \, \, \delta \mathfrak{v}_3= c,\\
&(B2):\, \,\, \, \delta \mathfrak{v}_1= a+b, \, \, \, \delta \mathfrak{v}_2= c,  \, \, \, \delta \mathfrak{v}_3= d,\\
&(B3):\, \,\, \, \delta \mathfrak{v}_1= d, \, \, \, \delta \mathfrak{v}_2= a+b,  \, \, \, \delta \mathfrak{v}_3= c,\\
&(B4):\, \,\, \, \delta \mathfrak{v}_1= c, \, \, \, \delta \mathfrak{v}_2= a+b,  \, \, \, \delta \mathfrak{v}_3= d,\\
&(B5):\, \,\, \, \delta \mathfrak{v}_1= c, \, \, \, \delta \mathfrak{v}_2= d,  \, \, \, \delta \mathfrak{v}_3= a+b,\\
&(B6):\, \,\, \, \delta \mathfrak{v}_1= d, \, \, \, \delta \mathfrak{v}_2= c,  \, \, \, \delta \mathfrak{v}_3= a+b.
\end{align*}

Taking any cases of $A$ or $B$, we have $\delta \mathfrak{t}_1 + \delta \mathfrak{t}_2 + \delta \mathfrak{t}_3 = \delta \mathfrak{v}_1 + \delta \mathfrak{v}_2 + \delta \mathfrak{v}_3 = a+b+c+d$. Comparing any cases of $A$ with any cases of $B$, $\delta \mathfrak{t}_1 \geq \delta \mathfrak{v}_1$. Taking case $(A1)$, since $a>b$ then we have $c+d+a \geq a+b+d$ and $c+d+a \geq a+b+c$ and also simply we have $c+d+a \geq c+d$. Therefor, Taking case $A1$ and comparing with any cases of $B$, we have $\delta \mathfrak{t}_1 + \delta \mathfrak{t}_2 \geq \delta \mathfrak{v}_1 + \delta \mathfrak{v}_2$. The same reasoning also can be done for case $A2$. Hence taking any cases of $A$ and $B$, we know that $\delta \mathfrak{t}_1,\delta \mathfrak{t}_2,\delta \mathfrak{t}_3$ majorizes $\delta \mathfrak{v}_1,\delta \mathfrak{v}_2,\delta \mathfrak{v}_3$. This results in 
\begin{align*}
     f(c)+f(d)+f(a+b)\leq f(a)+f(b)+f(c+d)
\end{align*}
and consequently 

\begin{align*}
    f(a)+f(b)-f(a+b)\leq f(c)+f(d)-f(c+d).
\end{align*}

\end{proof}

\begin{lem}\label{lem::auxMain1}
For any $(\mathfrak{q})_{1}^l$, let 
$$
     g((\mathfrak{q})_1^{l})=\sum\nolimits_{i=1}^{l-1}f(\Delta \mathfrak{q}_i),
$$
where $\Delta \mathfrak{q}_i=\mathfrak{q}_{i+1}-\mathfrak{q}_{i}$ and $f$ be a concave and increasing function with $f(0)=0$. Now, consider two increasing sequences $(\mathfrak{t})^{n}_1$ and $(\mathfrak{u})_1^{l}$, and their concatenation $(\mathfrak{a})_1^{n+l}=(\mathfrak{t})^{n}_1 \oplus (\mathfrak{u})_1^{l}$. Then,
$$
  g((\mathfrak{a})_1^{n+l})-g((\mathfrak{t})_1^{n}) \geq 0 .
$$ holds
\end{lem}

\begin{proof}
If $\mathfrak{a}_{p}\!=
\!t_1$ and $\mathfrak{a}_{q}\!=\!t_n$, then since~$(\mathfrak{a})_1^{n+l}$ is a increasing sequence, $p\!<\!q$. Let the sub-sequence of $(\mathfrak{a})_1^{n+l}$ ranging from index $p$ to $q$ be $(\mathfrak{v})_1^{m}$ where $m\geq n$. Letting $\Delta \mathfrak{v}_i=\mathfrak{v}_{i+1}-\mathfrak{v}_i$ and $\Delta \mathfrak{t}_i=\mathfrak{t}_{i+1}-\mathfrak{t}_i$, we rearrange $\Delta \mathfrak{v}_i$'s and $\Delta \mathfrak{t}_i$'s in a descending order to form the sequences $(\delta \mathfrak{v})_1^{l-1}$ and $(\delta \mathfrak{t})_1^{n-1}$. Since $\mathfrak{a}_{p}=\mathfrak{t}_1$ and $\mathfrak{a}_{q}=\mathfrak{t}_n$, we have
$$
    \SUM{i=1}{m-1}\Delta \mathfrak{v}_i = \SUM{i=1}{m-1}\delta \mathfrak{v}_i =
    \SUM{i=1}{n-1}\delta \mathfrak{t}_i = \SUM{i=1}{n-1}\Delta \mathfrak{t}_i = \mathfrak{t}_n - \mathfrak{t}_1.
$$
Because $(\mathfrak{a})_1^{n+l}=(\mathfrak{t})^{n}_1 \oplus (\mathfrak{u})_1^{l}$, then $\forall i\in \{1,\cdots, n \}$ there exists $\mathsf{S}_i\! \subset\! \{1,\cdots,m\}$ such that $\sum\nolimits_{j \in \mathsf{S}_i}\delta \mathfrak{v}_j = \delta \mathfrak{t}_i$, where $\mathsf{S}_i \cap \mathsf{S}_k =\emptyset, \,\, i \not =k $. Consequently, for $r\! \in\! \{1, \cdots,m\}$, we have $ \sum\nolimits_{i=1}^{r}\delta \mathfrak{v}_i \!=\! \sum\nolimits_{j \in \mathsf{S}}\delta \mathfrak{t}_j$ for $\mathsf{S} \!\subset\! \{1,\cdots,n\}$ and $|S|\leq r$. Since $(\delta \mathfrak{t})_1^{n-1}$ is a decreasing sequence, we can write
$$
    \sum\nolimits_{i=1}^{r}\delta \mathfrak{v}_i \leq \sum\nolimits_{i=1}^{r}\delta \mathfrak{t}_i.
$$
Thus,
$$
    f(\delta \mathfrak{t}_1)\!+\!\cdots\!+\!f(\delta \mathfrak{t}_{n-1}) \!\leq \!f(\delta \mathfrak{v}_1)\!+\!\cdots\!+\!f(\delta \mathfrak{v}_{m-1})
$$ holds as a result of Lemma \ref{lem::krm}.
Given that $$ f(\delta \mathfrak{t}_1)\!+\!\cdots\!+\!f(\delta \mathfrak{t}_{n-1})\! =\! \SUM{i=1}{n-1}f(\Delta \mathfrak{t}_i)$$ and $$ f(\delta \mathfrak{v}_1)\!+\!\cdots\!+\!f(\delta \mathfrak{v}_{m-1}) \!=\! \SUM{i=1}{m-1}f(\Delta \mathfrak{v}_i) \leq \SUM{i=1}{n+1-1}f(\Delta \mathfrak{a}_i),$$ then $$\SUM{i=1}{n-1}f(\Delta \mathfrak{t}_i) \!\leq\! \SUM{i=1}{n+1-1}f(\Delta \mathfrak{a}_i),$$ which concludes the~proof. 
\end{proof}

\begin{lem}\label{lem::auxMain2}
 For any $(\mathfrak{q})_{1}^l$, let $$g((\mathfrak{q})_1^{l})=\sum\nolimits_{i=1}^{l}f(\Delta \mathfrak{q}_i)$$
where $\Delta \mathfrak{q}_i=\mathfrak{q}_{i+1}-\mathfrak{q}_{i}$ and $f$ is a concave and increasing function with $f(0)=0$. Now, consider three increasing sequences $(\mathfrak{t})^{n}_1$ and $(\mathfrak{v})^{m}_1$ and $(\mathfrak{u})_1^{l}$ and concatenations $(\mathfrak{a})_1^{n+l}=(\mathfrak{t})^{n}_1 \oplus (\mathfrak{u})_1^{l}$ and $(\mathfrak{b})_1^{m+l}=(\mathfrak{v})^{m}_1 \oplus (\mathfrak{u})_1^{l}$ where $(\mathfrak{v})^{m}_1$ is a sub-sequence of $(\mathfrak{t})^{n}_1$, then
$$
   \big(g((\mathfrak{b})_1^{m+l})-g((\mathfrak{v})_1^{m})\big)- \big(g((\mathfrak{a})_1^{n+l})-g((\mathfrak{t})_1^{n})\big) \geq 0.
$$
\end{lem}

\begin{proof}
Let the sequence $(\mathfrak{u})_1^{p}$ be the first $p$ elements of $(\mathfrak{u})_1^{l}$. Then, we can form 
\begin{align*}
 \Delta S_p = \big(g((\mathfrak{v})_1^{m}\oplus (\mathfrak{u})_1^{p})
        &-g((\mathfrak{v})_1^{m}\oplus (\mathfrak{u})_1^{p-1})\big)\\
        &-\\
   \big(g((\mathfrak{t})_1^{n}\oplus (\mathfrak{u})_1^{p})
        &-g((\mathfrak{t})_1^{n}\oplus (\mathfrak{u})_1^{p-1})\big),
\end{align*}
 where $(\mathfrak{u})_1^{0}$ to be an empty sequence with no members. Since $(\mathfrak{v})^{m}_1$ is a sub-sequence of $(\mathfrak{t})^{n}_1$ and $(\mathfrak{u})_1^{p}$ having one member more over $(\mathfrak{u})_1^{p-1}$, then we have 
 \begin{align*}
 \Delta S_p = (f(\Delta S_{1})+f(\Delta S_{2})&-f(\Delta S_{1}+\Delta S_{2}))\\
   & - \\
    (f(\Delta S_{3})+f(\Delta S_{4})&-f(\Delta S_{3}+\Delta S_{4}))
 \end{align*}    
with $0\leq \Delta S_{3} \leq \Delta S_{1}$ and $0\leq \Delta S_{4} \leq \Delta S_{2}$. From Corollary~\ref{lm::aux1}, we can conclude that $\Delta S_p\geq 0$. Then, given 
$$
    \sum_{p=1}^{l} \!\Delta S_p \!\!=\!\!\big(g((\mathfrak{b})_1^{m+l})\!-\!g((\mathfrak{v})_1^{m})\big)\!\!-\! \big(g((\mathfrak{a})_1^{n+l})\!-g((\mathfrak{v})_1^{m})\big),
$$
 the proof is concluded.
\end{proof}

\end{document}